\documentclass[a4paper,UKenglish,cleveref, autoref, thm-restate]{lipics-v2021}
\nolinenumbers
\hideLIPIcs  %uncomment to remove references to LIPIcs series (logo, DOI, ...), e.g. when preparing a pre-final version to be uploaded to arXiv or another public repository
\usepackage{graphicx} % Required for inserting images
\usepackage{amsmath,amssymb,amsthm}
\usepackage[T1]{fontenc}
\usepackage{color,url}
\usepackage{hyperref}
\usepackage{xspace}
\usepackage{todonotes}
\usepackage{nicematrix}
\usepackage{lineno}
\usepackage{circuitikz}

\usepackage[noadjust,compress]{cite}

\usepackage[noend]{algpseudocode}
\usepackage{algorithm}

\title{A Polynomial Time Algorithm for Steiner Tree when Terminals Avoid a Rooted $\texorpdfstring{K_4}{K4}$-Minor}
\titlerunning{A Polynomial Time Algorithm for Steiner Tree when Terminals Avoid a $\texorpdfstring{K_4}{K4}$-Minor}
\author{Carla Groenland}{Delft Institute of Applied Mathematics}{C.E.Groenland@tudelft.nl}{https://orcid.org/0000-0002-9878-8750}{}
\author{Jesper Nederlof}{Utrecht University, Department of Information and Computing Sciences}{j.nederlof@uu.nl}{https://orcid.org/0000-0003-1848-0076}{}
\author{Tomohiro Koana}{Utrecht University, Department of Information and Computing Sciences}{tomohiro.koana@gmail.com}{https://orcid.org/0000-0002-8684-0611}{}

\acknowledgements{The work of both JN and TK has been supported by the project CRACKNP that has received funding from the European Research Council (ERC) under the European Union's Horizon 2020 research and innovation programme (grant agreement No 853234).}

\authorrunning{Groenland et al.}


 %TODO mandatory; please add comma-separated list of keywords

\Copyright{Carla Groenland, Jesper Nederlof, and Tomohiro Koana} %TODO mandatory, please use full first names. LIPIcs license is "CC-BY";  http://creativecommons.org/licenses/by/3.0/

\EventEditors{\'{E}douard Bonnet and Pawe\l{} Rz\k{a}\.{z}ewski}
\EventNoEds{2}
\EventLongTitle{19th International Symposium on Parameterized and Exact Computation (IPEC 2024)}
\EventShortTitle{IPEC 2024}
\EventAcronym{IPEC}
\EventYear{2024}
\EventDate{September 4--6, 2024}
\EventLocation{Royal Holloway, University of London, Egham, United Kingdom}
\EventLogo{}
\SeriesVolume{321}
\ArticleNo{12}

\DeclareMathOperator{\dist}{dist}

\date{\today}

\begin{document}

\maketitle
\begin{abstract}
    We study a special case of the Steiner Tree problem in which the input graph does not have a minor model of a complete graph on $4$ vertices for which all branch sets contain a terminal. We show that this problem can be solved in $O(n^4)$ time, where $n$ denotes the number of vertices in the input graph.
    This generalizes a seminal paper by Erickson et al.~[Math. Oper. Res., 1987] that solves Steiner tree on planar graphs with all terminals on one face in polynomial time.  
\end{abstract}

\section{Introduction}
Planar graphs are well-studied algorithmically. For example, starting with the work of Baker~\cite{DBLP:journals/jacm/Baker94}, many efficient approximation schemes for NP-complete problems on planar graphs have been designed. Within parameterized complexity, widely applicable tools such as bi-dimensionality~\cite{DBLP:journals/jacm/DemaineFHT05} helped to grasp a firm understanding of the ``square-root phenomenon'': Many problems can be solved in $2^{O(\sqrt{n})}$ time, or even in $2^{O(\sqrt{k})} \cdot n^{O(1)}$ time, where $n$ denotes the number of vertices of the input graph and $k$ denotes the size of the sought solution.
Actually most of these algorithms are also shown to extend to superclasses\footnote{Wagner's theorem states that a graph is planar unless it contains a complete graph on $5$ vertices ($K_5$) or complete bipartite graph with two blocks of $3$ vertices ($K_{3,3}$) as a minor.} of planar graphs, such as bounded genus or even minor free graphs (as for example showed in~\cite{DBLP:journals/jacm/DemaineFHT05}).

One problem that is very well-studied on planar graphs is \textsc{Steiner Tree}.
In the \textsc{Steiner Tree} problem we are given a weighted graph $G=(V,E)$ on $n$ vertices with edge weights $w_e \in \mathbb{R}_{\geq 0}$ for $e \in E$ and a set of vertices $T = \{ t_1,\dots,t_k \} \subseteq V$ called \emph{terminals}, and we are tasked with finding an edge set $S$ minimizing $w(S)=\sum_{e\in S}w(e)$ such that any pair of terminals $t, t'$ are connected in the subgraph $(V, S)$.
For example, an efficient approximation scheme was given in~\cite{DBLP:journals/talg/BorradaileKM09}, and a lower bound excluding $2^{o(k)}$ time algorithms on planar graphs under the Exponential Time Hypothesis was given in~\cite{DBLP:conf/focs/MarxPP18}. Interestingly, the latter result shows that planarity alone is not too helpful to solve \textsc{Steiner Tree} quickly, because it implies that the classic  $3^{k}n^{O(1)}$ dynamic programming algorithm for general \textsc{Steiner Tree}~\cite{DBLP:journals/networks/DreyfusW71} cannot be significantly improved.

A seminal paper by Erickson et al.~\cite{EricksonMV87}  shows that \textsc{Steiner Tree} on planar graphs is in $P$ if all terminals lie on one face in a planar embedding of $G$.
The study of the setting in which terminals lie in a few number of faces dates back all the way to the work of Ford and Fulkerson~\cite{ford1956maximal}, and has been the subject of many classic works (such as the Okamura-Seymour Theorem~\cite[Chapter 74]{schrijver2003combinatorial} or~\cite{DBLP:journals/siamcomp/MatsumotoNS85}). 
The algorithm~\cite{EricksonMV87} is used as subroutine in several other papers such as the aforementioned approximation scheme~\cite{DBLP:journals/talg/BorradaileKM09}, but also preprocessing algorithms~\cite{DBLP:journals/talg/PilipczukPSL18}.

Often the algorithms that exploit planarity or minor-freeness need to combine \emph{graph theoretic} techniques (such as a grid minor theorem in the case of bi-dimensionality) with \emph{algorithmic} perspective (such as divide and conquer or dynamic programming over tree decompositions).
For the \textsc{Steiner Tree} problem, and especially the algorithm from~\cite{EricksonMV87}, the graph theoretic techiques employed are intrinsically \emph{geometric}.
And indeed, many extensions of the algorithm, such as the one in which the terminals can be covered the $k$ outer faces of the graph~\cite{DBLP:journals/anor/BernB91}, or the setting in which the terminals can be covered by a ``path-convex region'' studied in \cite{DBLP:journals/siamcomp/Provan88}, all have a strong geometric flavor.

With this in mind, it is natural to ask whether there is an extension of the algorithm of~\cite{EricksonMV87} to minor free graphs.

\paragraph*{Rooted Minors.} We study a setting with \emph{rooted} minors. A graph $H$ is a \emph{minor} of a graph $G$ if it can be constructed from $G$ by removing vertices or edges, or contracting edges. In a more general setting, $G$ has a set of \emph{roots} $R \subseteq V(G)$ and there is a mapping $\pi$ that prescribes for each $v \in V(H)$ the set of roots $\pi(v) \subseteq R$. Then \emph{a rooted minor} is a minor that contracts $r$ into $v$, for every $v \in V(H)$ and $r \in R$ such that $r \in \pi(v)$.
Rooted minors play a central in Robertson and Seymour’s graph minor theory, and directly generalize the \textsc{Disjoint Paths} problem. Recently it was shown~\cite{korhonen2024minor} that rooted minors can be detected in $(m+n)^{1+o(1)}$ time, for fixed $|H|$ and $|R|$, improving over the quadratic time algorithm from~\cite{DBLP:journals/jct/KawarabayashiKR12}. 
A number of recent papers have studied rooted minors in their own right~\cite{Wollan08,Kawarabayashi04,JorgensenKawarabayashi07,WoodLinusson,hodor2024quickly, DBLP:journals/jgt/BohmeHKMS24}.
In particular, Fabila-Monroy and Wood \cite{Fabila-MonroyAndWood} gave a characterization of when a graph contains a $K_4$-minor rooted at four given vertices. Recently, links between rooted minors and `rooted' variants of treewidth, pathwidth and treedepth were given in~\cite{hodor2024quickly}.  

Motivated by the missing extension of the algorithm of~\cite{EricksonMV87} to a (non-geometric) setting of excluded minors, we propose studying the complexity of instances without minors rooted at the terminals.

\paragraph*{Our Result}
We will be interested in the closely related (but slightly different) setting of \textit{$R$-rooted $K_4$-minor}: 
In a graph $G=(V,E)$ with vertices $R\subseteq V$ (referred to as \emph{roots}), an \textit{$R$-rooted $K_4$-minor} is a collection of disjoint vertex sets $S_1,S_2,S_3,S_4\subseteq V$ such that $G[S_i]$ is connected and $S_i\cap R\neq \emptyset$ for all $i\in \{1,\dots,4\}$, and such that there is an edge from $S_i$ to $S_j$ for each $i\neq j$. 

\begin{theorem}\label{thm:st}
\textsc{Steiner Tree} without $K_4$ minor rooted at terminals can be solved in $O(n^4)$ time.
\end{theorem}

This generalizes the result from~\cite{EricksonMV87}: It is easily seen that a planar graph has no $K_4$-minor rooted at four vertices on the same face. On the other hand, it is easy to come up with example instances that have no $K_4$-minor and are not planar (see Figure~\ref{fig:instance} for such two such examples).

\begin{figure}
\centering
\begin{tikzpicture}[scale=2]
    \begin{scope}
    \node (a) at (0, 2) [circle, fill, red, inner sep=2pt] {};
    \node (b) at (2, 2) [circle, fill, red, inner sep=2pt] {};
    \node (c) at (2, 0) [circle, fill, red, inner sep=2pt] {};
    \node (d) at (0, 0) [circle, fill, red, inner sep=2pt] {};
    \node (x) at (1, 1) [circle, fill, inner sep=2pt] {};

    \node (ab1) at (0.8, 1.5) [circle, fill, inner sep=2pt] {};
    \node (ab2) at (1.2, 1.5) [circle, fill, inner sep=2pt] {};
    \node (bc1) at (1.5, 1.2) [circle, fill, inner sep=2pt] {};
    \node (bc2) at (1.5, 0.8) [circle, fill, inner sep=2pt] {};
    \node (cd1) at (0.8, 0.5) [circle, fill, inner sep=2pt] {};
    \node (cd2) at (1.2, 0.5) [circle, fill, inner sep=2pt] {};
    \node (da1) at (0.5, 0.8) [circle, fill, inner sep=2pt] {};
    \node (da2) at (0.5, 1.2) [circle, fill, inner sep=2pt] {};

    \draw (a) -- (b) -- (c) -- (d) -- (a);
    \draw (a) -- (x);
    \draw (b) -- (x);
    \draw (c) -- (x);
    \draw (d) -- (x);

    \draw (a) -- (ab1) -- (b); \draw (a) -- (ab2) -- (b); \draw (ab1) -- (ab2);
    \draw (b) -- (bc1) -- (c); \draw (b) -- (bc2) -- (c); \draw (bc1) -- (bc2);
    \draw (c) -- (cd1) -- (d); \draw (c) -- (cd2) -- (d); \draw (cd1) -- (cd2);
    \draw (d) -- (da1) -- (a); \draw (d) -- (da2) -- (a); \draw (da1) -- (da2);

    \foreach \name in {ab1, ab2, bc1, bc2, cd1, cd2, da1, da2} {
        \draw (x) -- (\name);
    }
    \end{scope}
    \begin{scope}[shift={(-3,0.13)}]
        \node (a) at (-1, 0) [circle, fill, inner sep=2pt, red] {};
        \node (b) at (0, 1.73) [circle, fill, inner sep=2pt, red] {};
        \node (c) at (1, 0) [circle, fill, inner sep=2pt, red] {};

        \node (x1) at (-0.2, 0.58) [circle, fill, inner sep=2pt] {};
        \node (x2) at (0.2, 0.58) [circle, fill, inner sep=2pt] {};

        \draw (a) -- (b) -- (c) -- (a); 
        \draw (x1) -- (x2);
        \foreach \name in {a, b, c} {
            \draw (x1) -- (\name); \draw (x2) -- (\name);
        }
    \end{scope}
\end{tikzpicture}
\caption{Terminals are depicted in red. The left figure has $3$ terminals and hence no rooted $K_4$ minor, but it is a $K_5$ and hence not planar. The right figure has no rooted $K_4$ minor (since the middle vertex and two non-adjacent terminals separate the remaining terminals), and has many $K_5$ subgraphs and hence is not planar.}
\label{fig:instance}
\end{figure}
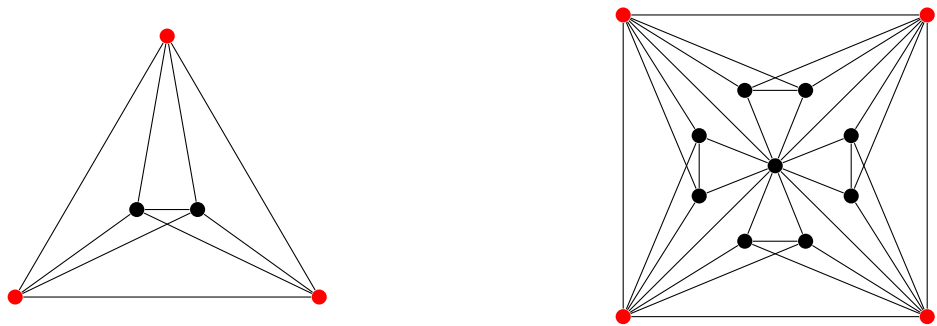

We hope that our result paves the road for additional polynomial time algorithms (solving non-geometrically) restricted versions of \textsc{Steiner Tree}. In particular, it would be interesting to see whether Theorem~\ref{thm:st} can be extended to a polynomial time algorithm for a richer set of rooted $\mathcal{F}$-minor free instances (e.g., instances that do not contain any member of $\mathcal{F}$ as rooted minor). Note however that such a strengthening with $\mathcal{F}=\{K_5\}$ would imply P=NP since \textsc{Steiner Tree} on planar graphs is NP-hard.

\paragraph*{Our Approach}
It is not too difficult to show from ideas given in \cite{Fabila-MonroyAndWood} that if there is no $R$-rooted $K_4$-minor in a 3-connected graph, then there is a cycle $C$ passing through all vertices in $R$. We need a slightly stronger variant of this result, where the graph is not quite 3-connected but the only 2-cuts allowed are those which isolate a single vertex that is a terminal (see Lemma \ref{lem:comb_lemma}). Our algorithm will recurse on 2-cuts until it can apply this lemma, and then performs dynamic programming in a similar fashion to the Dreyfus-Wagner algorithm~\cite{DBLP:journals/networks/DreyfusW71} restricted to subsets of terminals that form a contiguous segment of the cycle $C$ (in a fashion that is analogous to the algorithm for~\cite{EricksonMV87}, although the virtual edges considerably increase the technical difficulty of our proof). Since the number of such segments is at most quadratic in $n$, this gives a polynomial time algorithm.

We stress that the approach of~\cite{EricksonMV87} does not work directly: The approach of~\cite{EricksonMV87} crucially relies on the fact that if one decomposes an optimal Steiner Tree $S$ into two trees $S_1 \cup S_2$ where $S_1,S_2$ are the connected components of $S \setminus e$ for some $e \in S$, then the set $R_1$ (and $R_2$) of terminals covered by $S_1$ (and $S_2$) is a contiguously visited along the cycle enclosing the face that contains all terminals. In our setting, there is a priori no guarantee how the set $R_1$ will look, and therefore such a decomposition approach will not work. To overcome this, we also decompose the optimal tree along $2$-cuts via recursion and processing outcomes of recursive calls with virtual edges. Even though the virtual edges lead to some technical challenges, this allows us to make the idea of~\cite{EricksonMV87} work in $3$-connected graphs.

\section{Preliminaries}
For \( n \in \mathbb{N} \), we define \([n]\) as the set \(\{ 1, \dots, n \}\).
In this work, we assume that all graphs do not have self-loops, but we allow the graph to have parallel edges.
We will say that a graph is \emph{simple} if there is no parallel edge.
The vertex set and edge set of a graph $G$ are denoted by $V(G)$ and $E(G)$, respectively. 
% We denote the number of vertices by \( n := |V(G)| \).

For two vertices $u$ and $v$, let $\dist(u, v)$ be the length of a shortest path between $u$ and $v$.
Let \( X \subseteq V(G) \) be a subset of vertices.
We use \( G[X] \) to denote the subgraph induced by \( X \) and \( G - X := G[V(G) \setminus X] \) to represent the graph obtained by removing the vertices in \( X \).
For these notations, if \( X \) is a singleton \(\{ x \}\), we may write \( x \) instead of \(\{x\}\).
Moreover, $X$ is called a \emph{cut} if deleting $X$ increases the number of connected components. 
In particular, $x$ is called a \emph{cut vertex} if $\{ x \}$ is a cut.
For an edge set $F \subseteq E$, let $V(F)$ denote the set of vertices \emph{covered} by $F$, i.e., $V(F) = \{ u\mid u\in e \in F \}$.

% For a tree $\mathcal{T} =(V(\mathcal{T}), E(\mathcal{T}))$ and a vertex $v \in V(\mathcal{T})$, we say that a subtree $\mathcal{T}'$ of $\mathcal{T}$ is \emph{anchored at $v$} if it is obtained from $\mathcal{T}$ by  $\mathcal{T}' = \mathcal{T}[v \cup \bigcup_{i \in C} V(\mathcal{T}_i)]$, where
% $\mathcal{T}_1, \dots, \mathcal{T}_{\ell}$ are the connected components of $\mathcal{T} - v$ and $C \subseteq [\ell]$.
% The subtree $\mathcal{T}[v \cup \bigcup_{i \in [\ell] \setminus C} V(\mathcal{T}_i)]$ (which is also anchored at $v$) is denoted by $\mathcal{T} \setminus \mathcal{T}'$.
% In this work, we sometimes refer to a nonempty edge set $S \subseteq E$ as a tree if the graph $(V, S)$ has no cycle and has exactly one connected component of size greater than one.

Avoiding a (rooted) minor is preserved under deletion of vertices, deletion of edges and contraction of edges.
\begin{observation}
If $G$ does not contain an $R$-rooted $K_4$-minor and $H$ is a minor of $G$, then $H$ does not contain an $(V(H)\cap R)$-rooted $K_4$-minor.  
\end{observation}

Our structural analysis will also crucially build on the following simple lemma.

\begin{lemma}[Lemma 7 in \cite{Fabila-MonroyAndWood}] \label{lemma:cycle-to-k4}
    Suppose that a graph $G$ has a cycle containing vertices $v_1, v_2, v_3, v_4$ in that order.
    Suppose moreover that there are two disjoint paths $P_1$ and $P_2$ where $P_1$ is from $v_1$ to $v_3$, and $P_2$ is from $v_2$ to $v_4$.
    Then $G$ has a $\{v_1,v_2,v_3,v_4\}$-rooted $K_4$-minor.
\end{lemma}

\section{Finding a cycle passing through terminals and virtual edges}
\label{sec:comb_lemma}
In this section, we prove a structural lemma that our algorithm needs. This builds on ideas from~\cite{Fabila-MonroyAndWood}, but we require additional analysis due to the presence of `virtual edges'.
\begin{lemma}
\label{lem:comb_lemma}
    Let $G=(V,E)$ be a 3-connected graph and let $R\subseteq V$ be a set of roots with $|R|\geq 3$. 
    Let $E'\subseteq E$ and let $G'$ be obtained by subdividing each edge in $E'$ once. Let $S$ denote the set of subdivision vertices added with this operation. 
    
    If $G'$ has no $(R\cup S)$-rooted $K_4$-minor, then $G'$ contains a cycle containing all vertices in $R\cup S$.
    Furthermore, this cycle can be found in $n m^{1+o(1)}$ time for $n$ the number of vertices and $m$ the number of edges of $G'$.
\end{lemma}
\begin{proof}
We first show that if there is no $R$-rooted $K_4$-minor in $G$ and $G$ is 3-connected, then we can find a cycle $C$ in $G$ that passes through all vertices of $R$.

We start with $C$ being any cycle.
%Let $C$ be a cycle which maximises the number of vertices of $R$ contained in it. 
%In a 3-connected graph, any three vertices lie on a common cycle \cite{Dirac}, so we may assume that $C$ contains at least three vertices from $R$.
Suppose there is a vertex $r\in R$ that is not on $C$. Since $G$ is 3-connected, by Menger's theorem (see~\cite[Theorem 3.3.1]{diestel2017graph}, applied with sets $C$ and $N(v)$), there exist three paths $P_1,P_2,P_3$ from $r$ to $C$, where the paths mutually only intersect in $r$ and each path only intersects $C$ in their other endpoint.
We can find these paths in $m^{1+o(1)}$ time with the max flow algorithm from e.g.~\cite{DBLP:conf/stoc/0028KLMG24}.
Let $v_1,v_2,v_3\in C$ denote these endpoints.
If there are three disjoint arcs on the cycle that each contains a root from $R$ and one of $\{v_1,v_2,v_3\}$, then $G$ contains a rooted $K_4$ minor (see Figure \ref{fig:threesegments}).

\begin{figure}[h!]
    \centering
\resizebox{1\textwidth}{!}{%
\begin{circuitikz}
\tikzstyle{every node}=[font=\LARGE]
\draw [gray, line width=30pt,  opacity =0.5, line cap=round]  (-5.25,19.75) .. controls (-4.75,21) and (-3.5,21.75) .. (-2.25,21.5);
\draw [gray, line width=30pt,  opacity =0.5, line cap=round] (-5.5,18.5) .. controls (-4,14.5) and (-0.75,16.25) .. (-0.25,17.75);
\draw [gray, line width=30pt,  opacity =0.5, line cap=round] (-1.25,21) .. controls (-0.25,20.25) and (-0.25,20) .. (0,19.25);
\draw [short] (8,21.25) .. controls (10.75,23) and (10.75,23.25) .. (13.5,21.25);
\draw [green, line width=10pt,  opacity =0.5, line cap=round] (8.5,21.7) .. controls (10.75,23) and (10.75,23.25) .. (13.5,21.25);
\draw [green, line width=10pt,  opacity =0.5, line cap=round] (10.7,19.3) -- (13.5,21.25);
\draw [short] (10.25,19) -- (13.5,21.25);
\draw [green, line width=10pt,  opacity =0.5, line cap=round] (9.5,16.3) .. controls (12.75,17.5) and (13.5,17.5) .. (13.5,21.25);
\draw [short] (8.75,16) .. controls (12.75,17.5) and (13.5,17.5) .. (13.5,21.25);
\draw  (7.5,18.5) circle (2.75cm);
\draw [ fill={rgb,255:red,0; green,0; blue,0} ] (8.75,16) circle (0.25cm); % bottom
\draw [yellow, line width=10pt,  opacity =0.5, line cap=round]  (7.25,15.8) .. controls (9,16) and (9.5,16.5) .. (10.2,18);
\draw [ color={rgb,255:red,0; green,0; blue,0} , fill={rgb,255:red,255; green,255; blue,255}] (9.75,17.75) rectangle (10.5,17); % bottom
\draw [ fill={rgb,255:red,0; green,0; blue,0} ] (8,21.25) circle (0.25cm); % top
\draw [blue, line width=10pt,  opacity =0.5, line cap=round]  (8.5,21.15) .. controls (3.75,21.25) and (4,17) .. (6.75,16);
\draw [ color={rgb,255:red,0; green,0; blue,0} , fill={rgb,255:red,255; green,255; blue,255}] (6.25,21.5) rectangle (7,20.75); % top
\draw [ fill={rgb,255:red,0; green,0; blue,0} ] (10.25,19) circle (0.25cm); % middle
\draw [red, line width=10pt,  opacity =0.5, line cap=round]  (10.25,18.8) .. controls (10,20) and (9.5,20.5) .. (9,20.8);
\draw [ color={rgb,255:red,0; green,0; blue,0} , fill={rgb,255:red,255; green,255; blue,255}] (9,20.75) rectangle (9.75,20); % middle
\node [font=\LARGE] at (10.75,18.5) {};
\node [font=\LARGE] at (10.75,18.5) {};
\node [font=\LARGE] at (8.5,20.75) {};
\draw [short] (-2.25,21.5) .. controls (0.5,23.25) and (0.5,23.5) .. (3.25,21.5);
\draw [short] (0,19.25) -- (3.25,21.5);
\draw [ fill={rgb,255:red,0; green,0; blue,0} ] (0,19.25) circle (0.25cm);
\draw [ fill={rgb,255:red,0; green,0; blue,0} ] (-2.25,21.5) circle (0.25cm);
\draw [short] (-1.5,16.25) .. controls (2.5,17.75) and (3.25,17.75) .. (3.25,21.5);
\node [font=\LARGE] at (0.5,18.75) {};
\node [font=\LARGE] at (0.5,18.75) {};
\node [font=\LARGE] at (-1.75,21) {};
\draw [ fill={rgb,255:red,0; green,0; blue,0} ] (-1.5,16.25) circle (0.25cm);
\draw  (-2.75,18.75) circle (2.75cm);
\draw [ color={rgb,255:red,1; green,0; blue,0} , fill={rgb,255:red,254; green,255; blue,255}] (-4,21.75) rectangle (-3.25,21);
\draw [ color={rgb,255:red,1; green,0; blue,0} , fill={rgb,255:red,254; green,255; blue,255}] (-0.5,18) rectangle (0.25,17.25);
\draw [ color={rgb,255:red,1; green,0; blue,0} , fill={rgb,255:red,254; green,255; blue,255}] (-1.25,21) rectangle (-0.5,20.25);
\draw [ color={rgb,255:red,1; green,0; blue,0} , fill={rgb,255:red,254; green,255; blue,255}] (13,21.5) rectangle (13.75,20.75);
\draw [ color={rgb,255:red,1; green,0; blue,0} , fill={rgb,255:red,254; green,255; blue,255}] (2.75,21.75) rectangle (3.5,21);
\end{circuitikz}
}%    
\caption{If there are three vertex-disjoint paths from a root to disjoint arcs of the cycle containing roots (left), we obtain a rooted $K_4$-minor (right). The roots are depicted by boxes.}
    \label{fig:threesegments}
\end{figure}
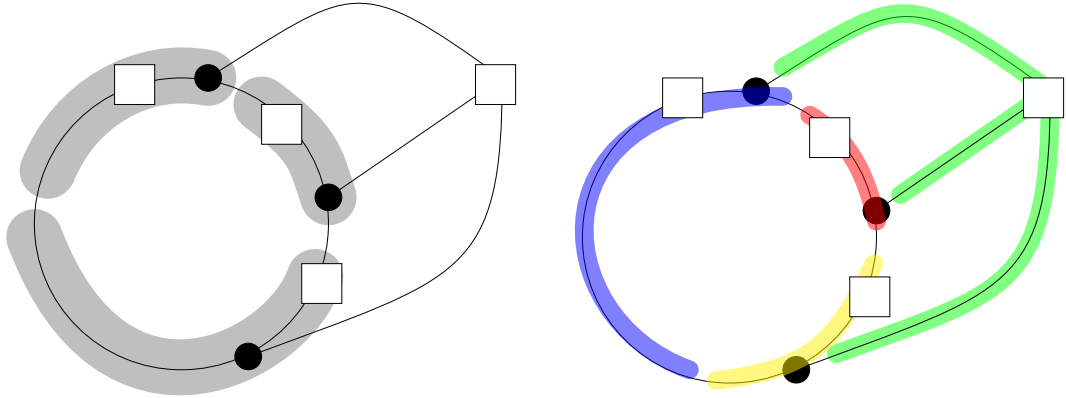

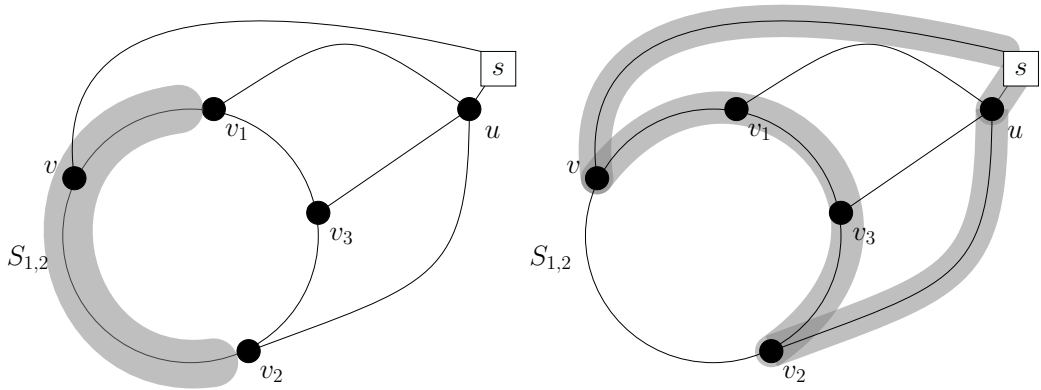
\begin{figure}[!ht]
\centering
\resizebox{1\textwidth}{!}{%
\begin{circuitikz}
\tikzstyle{every node}=[font=\LARGE]
\draw [short] (3.5,12.5) -- (4,13.25);
\draw [short] (-2,12.5) .. controls (0.75,14.25) and (0.75,14.5) .. (3.5,12.5);
\draw [short] (0.25,10.25) -- (3.5,12.5);
\draw [ fill={rgb,255:red,0; green,0; blue,0} ] (0.25,10.25) circle (0.25cm) node {\LARGE $v_3$} ;
\draw [ fill={rgb,255:red,0; green,0; blue,0} ] (-2,12.5) circle (0.25cm);
\draw [short] (-1.25,7.25) .. controls (2.75,8.75) and (3.5,8.75) .. (3.5,12.5);
\node [font=\LARGE] at (0.75,9.75) {};
\node [font=\LARGE] at (0.75,9.75) {};
\node [font=\LARGE] at (-1.5,12) {};
\draw [ fill={rgb,255:red,0; green,0; blue,0} ] (-1.25,7.25) circle (0.25cm);
\draw  (-2.5,9.75) circle (2.75cm);
\draw [ color={rgb,255:red,1; green,0; blue,0} , fill={rgb,255:red,254; green,255; blue,255}] (3.75,13.75) rectangle  node {\LARGE $s$} (4.5,13);
\draw [gray, line width=30pt,  opacity =0.5, line cap=round] (-2.75,12.5) .. controls (-6.5,12) and (-5.5,6.5) .. (-2,7);
\draw [ fill={rgb,255:red,0; green,0; blue,0} ] (3.5,12.5) circle (0.25cm);
\node [font=\LARGE] at (4,12) {$u$};
\node [font=\LARGE, rotate around={-360:(0,0)}] at (-1.5,12) {$v_1$};
\node [font=\LARGE] at (0.75,9.75) {$v_3$};
\node [font=\LARGE] at (-0.75,6.75) {$v_2$};
\draw [ fill={rgb,255:red,0; green,0; blue,0} ] (-5,11) circle (0.25cm) node {\LARGE $v$} ;
\draw [short] (3.75,13.75) .. controls (-4.75,15.75) and (-5.25,12.75) .. (-5,11);
\node [font=\LARGE] at (-6,9.25) {$S_{1,2}$};
\node [font=\LARGE] at (-5.5,11.25) {$v$};
\draw [short] (14.75,12.5) -- (15.25,13.25);
\draw [gray, line width=20pt,  opacity =0.5, line cap=round] (15,13.75) .. controls (6.5,15.75) and (6,12.75) .. (6.25,11);
\draw [gray, line width=20pt,  opacity =0.5, line cap=round] (6.25,11) .. controls (9.25,15) and (14.25,10.25) .. (10,7.25);
\draw [gray, line width=20pt,  opacity =0.5,line cap=round] (14.75,12.5) -- (15.25,13.25);
\draw [gray, line width=20pt,  opacity =0.5] (10,7.25) .. controls (14,8.75) and (14.75,8.75) .. (14.75,12.5);

\draw [short] (9.25,12.5) .. controls (12,14.25) and (12,14.5) .. (14.75,12.5);
\draw [short] (11.5,10.25) -- (14.75,12.5);
\draw [ fill={rgb,255:red,0; green,0; blue,0} ] (11.5,10.25) circle (0.25cm) node {\LARGE $v_3$} ;
\draw [ fill={rgb,255:red,0; green,0; blue,0} ] (9.25,12.5) circle (0.25cm);
\node [font=\LARGE] at (12,9.75) {};
\node [font=\LARGE] at (12,9.75) {};
\node [font=\LARGE] at (9.75,12) {};
\draw [short] (10,7.25) .. controls (14,8.75) and (14.75,8.75) .. (14.75,12.5);
\draw [ fill={rgb,255:red,0; green,0; blue,0} ] (10,7.25) circle (0.25cm);
\draw  (8.75,9.75) circle (2.75cm);
\draw [ color={rgb,255:red,1; green,0; blue,0} , fill={rgb,255:red,254; green,255; blue,255}] (15,13.75) rectangle  node {\LARGE $s$} (15.75,13);
\draw [ fill={rgb,255:red,0; green,0; blue,0} ] (14.75,12.5) circle (0.25cm);
\node [font=\LARGE] at (15.25,12) {$u$};
\node [font=\LARGE, rotate around={-360:(0,0)}] at (9.75,12) {$v_1$};
\node [font=\LARGE] at (12,9.75) {$v_3$};
\node [font=\LARGE] at (10.5,6.75) {$v_2$};
\draw [ fill={rgb,255:red,0; green,0; blue,0} ] (6.25,11) circle (0.25cm) node {\LARGE $v$} ;
\draw [short] (15,13.75) .. controls (6.5,15.75) and (6,12.75) .. (6.25,11);

\node [font=\LARGE] at (5.25,9.25) {$S_{1,2}$};
\node [font=\LARGE] at (5.75,11.25) {$v$};

\end{circuitikz}
}%
\caption{If the segment $S_{1,2}$ contains $v$ and no roots of $C$ (left), then there is a cycle through $v,u,s$ and all roots of $C$.}
\label{fig:cycle_with_s}
\end{figure}

But if such arcs do not exist, then after renumbering we may assume that there is no root vertex between $v_1$ and $v_2$ on $C$. The cycle $C'$ which goes from $v_1$ to $v_2$ via the path contained in $C$ containing $v_3$, and then back to $v_1$ via $r$ via $P_1$ and $P_3$, forms a cycle containing $(R\cap V(C))\cup \{r\}$. Hence we can reiterate with this cycle until $C$ contains $R$.

This means there is also a cycle $C$ in $G'$ containing all vertices in $R$: whenever an edge in $E'$ is used, it is possible to pass through the subdivision vertex instead. We next modify the cycle $C$ to be a cycle containing all vertices in $R$ and all vertices in $S$.

Let $uv\in E'$ be the edge whose subdivision resulted in $s$. 
We first ensure that $C$ contains both $u$ and $v$ (and $R$ and all vertices in $S$ already contained in it). 

Suppose that $u$ is not in $C$. As before, since $G$ is 3-connected and $u\in V(G)$, there must be three paths from $u$ to $C$, vertex-disjoint except for their endpoint $u$, meeting $C$ in distinct endpoints $v_1,v_2,v_3$. Again, we can find these paths in $m^{1+o(1)}$ time with~\cite{DBLP:conf/stoc/0028KLMG24}. We consider the same two cases as before.
\begin{itemize}
    \item \textit{Case 1: there are three disjoint arcs on $C$, each of which contains a vertex of $R\cup S$ and one vertex from $\{v_1,v_2,v_3\}$.} After contracting the edge between $s$ and $u$, we obtain a rooted $K_4$-minor in the same way as in the 3-connected case (see Figure \ref{fig:threesegments}). Hence this case does not happen by assumption.

    \item \textit{Case 2: the segment $S_{1,2}$ between $v_1$ and $v_2$ (excluding $v_1,v_2$ and $v_3$) contains no vertices from $S\cup R$.}    
    Then there is a cycle $C'$ containing $V(C)\cap (S\cup R)$ and $u$. Note that $C'$ no longer contains the vertices from $S_{1,2}$, so we could potentially lose $v$ if $v\in S_{1,2}$. But if $v\in S_{1,2}$, then there is a cycle $C''$ that goes through $V(C)\cap (S\cup R)\cup\{s\}$ (see Figure \ref{fig:cycle_with_s}). Either way we ensured that the cycle contains $u$ without affecting whether it contains $v$.    
\end{itemize}
After renumbering if needed, we are always either in Case 1 or Case 2.
Repeating the argument above for $v$ if needed, we can find a cycle passing through $(R\cup S)\cap V(C)\cup\{u,v\}$.

So we will assume $C$ passes through $u$ and $v$.
Now we claim that one of the two arcs of $C$ between $u,v\in C$ does not contain any vertices from $S\cup R$, and hence we may replace this by the path via $s$ to get a cycle $C$ that additionally passes through $s$.

Suppose this is not the case. 
Recall we assumed in our lemma statement that there are at least three vertices in $R$, and already established that $C$ contains all vertices of $R$.
So we can find $w,x,y$ in $R\cup S$ such that $C$ passes in order through $u, w, v,y,x,u$ (with some vertices in between those, possibly).  We choose $x,y$ such that there are no vertices from $R\cup S$ between $v$ and $y$ and between $x$ and $u$. (See also Figure \ref{fig:StoXY}.) 

Let $S_{uv}$ be the maximum segment of $C$ between $u$ and $v$ which includes $w$ but excludes $u$ and $v$. That is, if $u,a_1,\dots,a_\ell,w,b_1,\dots,b_k,v$ are the vertices of $C$ between $u$ and $v$ containing $w$, then
\[
S_{uv}=\{a_1,\dots,a_\ell,w,b_1,\dots,b_k\}.
\]
If there is a path intersecting $C$ only in its endpoints from a vertex in $S_{uv}$ to $x$ or the segment of $C$ in between $x,y$ that excludes $S_{uv}$, then we have an $(S\cup R)$-rooted $K_4$-minor in $G'$, as shown in Figure \ref{fig:StoXY}.
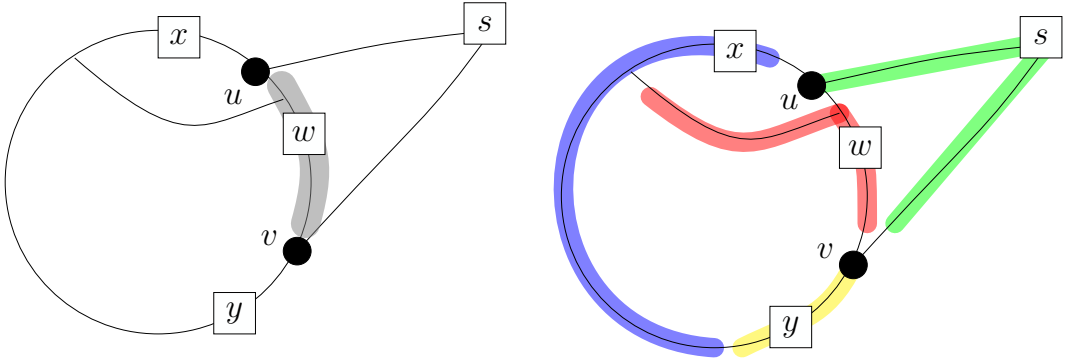
\begin{figure}
    \centering
\resizebox{1\textwidth}{!}{%
\begin{circuitikz}
\tikzstyle{every node}=[font=\LARGE]
\draw [short] (-0.75,11.75) .. controls (1.5,12.25) and (1.25,12.25) .. (3.5,12.5);
\draw [short] (0,8.5) .. controls (2.25,11) and (2.25,10.75) .. (3.5,12.5);
\node [font=\LARGE] at (0.75,9.75) {};
\node [font=\LARGE] at (0.75,9.75) {};
\node [font=\LARGE] at (-1.5,12) {};
\draw  (-2.5,9.75) circle (2.75cm);
\draw [ color={rgb,255:red,1; green,0; blue,0} , fill={rgb,255:red,254; green,255; blue,255}] (3,13) rectangle  node {\LARGE $s$} (3.75,12.25);
\draw [gray, line width=15pt,  opacity =0.5, line cap=round] (-0.28,11.5) .. controls (0.25,10.5) and (0.5,10.25) .. (0.15,9);
\draw [ fill={rgb,255:red,0; green,0; blue,0} ] (-0.75,11.75) circle (0.25cm);
\node [font=\LARGE] at (-1.15,11.3) {$u$};
\draw [ fill={rgb,255:red,0; green,0; blue,0} ] (0,8.5) circle (0.25cm) node {\LARGE $v$} ;
\node [font=\LARGE] at (-0.5,8.75) {$v$};
\draw [ color={rgb,255:red,1; green,0; blue,0} , fill={rgb,255:red,254; green,255; blue,255}] (-2.5,12.75) rectangle  node {\LARGE $x$} (-1.75,12);

\draw [yellow, line width=10pt,  opacity =0.5, line cap=round] (8,6.75) .. controls (9,7.25) and (9.5,7.25) .. (10,8.25);
\draw [blue, line width=10pt,  opacity =0.5, line cap=round] (8.5,12) .. controls (4,13.5) and (3.5,7) .. (7.5,6.75);
\draw [green, line width=10pt,  opacity =0.5, line cap=round]  (9.25,11.5) -- (13.25,12.25);
\draw [green, line width=10pt,  opacity =0.5, line cap=round]  (10.75,9) -- (13.5,12.25);
\draw [red, line width=10pt,  opacity =0.5, line cap=round]  (9.75,11) .. controls (8.25,10.5) and (8,10) .. (6.35,11.3);
\draw [red, line width=10pt,  opacity =0.5, line cap=round] (9.75,11) .. controls (10.25,10.25) and (10.25,10.25) .. (10.25,9);

\draw [ color={rgb,255:red,1; green,0; blue,0} , fill={rgb,255:red,254; green,255; blue,255}] (-1.5,7.75) rectangle  node {\LARGE $y$} (-0.75,7);
\draw [ color={rgb,255:red,1; green,0; blue,0} , fill={rgb,255:red,254; green,255; blue,255}] (-0.25,11) rectangle  node {\LARGE $w$} (0.5,10.25);
\draw [short] (-0.25,11.25) .. controls (-1.75,10.75) and (-2,10.25) .. (-4,12);

\draw [ fill={rgb,255:red,0; green,0; blue,0} , line width=0.2pt ] (0,11.5) circle (0cm);
\draw [short] (9.25,11.5) .. controls (11.5,12) and (11.25,12) .. (13.5,12.25);
\draw [short] (10,8.25) .. controls (12.25,10.75) and (12.25,10.5) .. (13.5,12.25);
\node [font=\LARGE] at (10.75,9.5) {};
\node [font=\LARGE] at (10.75,9.5) {};
\node [font=\LARGE] at (8.5,11.75) {};
\draw  (7.5,9.5) circle (2.75cm);
\draw [ color={rgb,255:red,1; green,0; blue,0} , fill={rgb,255:red,254; green,255; blue,255}] (13,12.75) rectangle  node {\LARGE $s$} (13.75,12);
\draw [ fill={rgb,255:red,0; green,0; blue,0} ] (9.25,11.5) circle (0.25cm);
\node [font=\LARGE] at (8.85,11.2) {$u$};
\node [font=\LARGE] at (10.75,9.5) {$
$};

\draw [ fill={rgb,255:red,0; green,0; blue,0} ] (10,8.25) circle (0.25cm) node {\LARGE $v$} ;
\node [font=\LARGE] at (9.5,8.5) {$v$};
\draw [ color={rgb,255:red,1; green,0; blue,0} , fill={rgb,255:red,254; green,255; blue,255}] (7.5,12.5) rectangle  node {\LARGE $x$} (8.25,11.75);
;
\draw [ color={rgb,255:red,1; green,0; blue,0} , fill={rgb,255:red,254; green,255; blue,255}] (8.5,7.5) rectangle  node {\LARGE $y$} (9.25,6.75);
\draw [ color={rgb,255:red,1; green,0; blue,0} , fill={rgb,255:red,254; green,255; blue,255}] (9.75,10.75) rectangle  node {\LARGE $w$} (10.5,10);
\draw [short] (9.75,11) .. controls (8.25,10.5) and (8,10) .. (6,11.75);

\draw [ fill={rgb,255:red,0; green,0; blue,0} , line width=0.2pt ] (10,11.25) circle (0cm);
\end{circuitikz}
}%
\caption{The segment $S_{uv}$ contains vertices between $u$ and $v$, including $w$ and excluding $u$ and $v$ itself (left). If there is a path from $S_{uv}$ to $x$ or the segment between $x$ and $y$, then there is a rooted $K_4$-minor (right), where $v$ is included with $y$ and $u$ with $s$.}
    \label{fig:StoXY}
\end{figure}

Similarly, if there is a path from $S_{uv}$ to $y$, a vertex in the segment between $x$ and $u$ (not including $u$) or the segment between $y$ and $v$ (not including $v$), then there is an $(S\cup R)$-rooted $K_4$-minor in $G'$, as depicted in Figure \ref{fig:StoXU}.

\begin{figure}[!ht]
\centering
\resizebox{1\textwidth}{!}{%
\begin{circuitikz}
\tikzstyle{every node}=[font=\LARGE]
\draw [yellow, line width=10pt,  opacity =0.5, line cap=round] (8,6.75) .. controls (8.75,7) and (9,7.25) .. (9.5,7.75);
\draw [blue, line width=10pt,  opacity =0.5, line cap=round] (9.25,11.5) .. controls (5.75,14.75) and (2,7.5) .. (7.5,6.75);
\draw [green, line width=10pt,  opacity =0.5, line cap=round] (10.25,11.75) .. controls (11.75,12.25) and (11.75,12) .. (13.25,12.25);
\draw [green, line width=10pt,  opacity =0.5, line cap=round] (10,8.25) -- (13.5,12.25);
\draw [red, line width=10pt,  opacity =0.5, line cap=round] (9.75,11) .. controls (8,9) and (8.75,9.5) .. (8.75,7.75);
\draw [red, line width=10pt,  opacity =0.5, line cap=round] (9.75,11) .. controls (10.25,10.25) and (10.25,10.25) .. (10.25,9);
\draw [yellow, line width=10pt,  opacity =0.5, line cap=round] (-2.75,6.75) .. controls (-1.75,7.25) and (-1.25,7.25) .. (-0.75,8.25);
\draw [blue, line width=10pt,  opacity =0.5, line cap=round] (-2.25,12) .. controls (-6.75,13.5) and (-7.25,7) .. (-3.25,6.75);
\draw [ green, line width=10pt,  opacity =0.5, line cap=round] (-1.5,11.5) -- (2.5,12.25);
\draw [green, line width=10pt,  opacity =0.5, line cap=round] (0,9) -- (2.75,12.25);
\draw [ red, line width=10pt,  opacity =0.5, line cap=round] (-1,11) .. controls (-0.5,10.25) and (-0.5,10.25) .. (-0.5,9);
\draw [ red, line width=10pt,  opacity =0.5, line cap=round] (-1,11) .. controls (-1.5,10.25) and (-2.75,10.75) .. (-2.5,11.25);
\draw [short] (9.25,11.5) .. controls (11.5,12) and (11.25,12) .. (13.5,12.25);
\draw [short] (10,8.25) .. controls (12.25,10.75) and (12.25,10.5) .. (13.5,12.25);
\node [font=\LARGE] at (10.75,9.5) {};
\node [font=\LARGE] at (10.75,9.5) {};
\node [font=\LARGE] at (8.5,11.75) {};
\draw  (7.5,9.5) circle (2.75cm);
\draw [ color={rgb,255:red,1; green,0; blue,0} , fill={rgb,255:red,254; green,255; blue,255}] (13,12.75) rectangle  node {\LARGE $s$} (13.75,12);

\draw [ fill={rgb,255:red,0; green,0; blue,0} ] (9.25,11.5) circle (0.25cm);
\node [font=\LARGE] at (8.75,11) {$u$};

\draw [ fill={rgb,255:red,0; green,0; blue,0} ] (10,8.25) circle (0.25cm) node {\LARGE $v$} ;
\node [font=\LARGE] at (9.5,8.5) {$v$};
\draw [ color={rgb,255:red,1; green,0; blue,0} , fill={rgb,255:red,254; green,255; blue,255}] (7.5,12.5) rectangle  node {\LARGE $x$} (8.25,11.75);

\draw [ color={rgb,255:red,1; green,0; blue,0} , fill={rgb,255:red,254; green,255; blue,255}] (8.5,7.5) rectangle  node {\LARGE $y$} (9.25,6.75);
\draw [ color={rgb,255:red,1; green,0; blue,0} , fill={rgb,255:red,254; green,255; blue,255}] (9.75,10.75) rectangle  node {\LARGE $w$} (10.5,10);
\draw [short] (9.75,11) .. controls (8,9.25) and (8.75,9.5) .. (8.75,7.25);

\draw [ fill={rgb,255:red,0; green,0; blue,0} , line width=0.2pt ] (10,11.25) circle (0cm);

\draw [short] (-1.5,11.5) .. controls (0.75,12) and (0.5,12) .. (2.75,12.25);
\draw [short] (-0.75,8.25) .. controls (1.5,10.75) and (1.5,10.5) .. (2.75,12.25);
\node [font=\LARGE] at (0,9.5) {};
\node [font=\LARGE] at (0,9.5) {};
\node [font=\LARGE] at (-2.25,11.75) {};
\draw  (-3.25,9.5) circle (2.75cm);
\draw [ color={rgb,255:red,1; green,0; blue,0} , fill={rgb,255:red,254; green,255; blue,255}] (2.25,12.75) rectangle  node {\LARGE $s$} (3,12);
\draw [ fill={rgb,255:red,0; green,0; blue,0} ] (-1.5,11.5) circle (0.25cm);
\node [font=\LARGE] at (-2,11) {$u$};

\draw [ fill={rgb,255:red,0; green,0; blue,0} ] (-0.75,8.25) circle (0.25cm) node {\LARGE $v$} ;
\node [font=\LARGE] at (-1.25,8.5) {$v$};
\draw [ color={rgb,255:red,1; green,0; blue,0} , fill={rgb,255:red,254; green,255; blue,255}] (-3.25,12.5) rectangle  node {\LARGE $x$} (-2.5,11.75);
\draw [ color={rgb,255:red,1; green,0; blue,0} , fill={rgb,255:red,254; green,255; blue,255}] (-2.25,7.5) rectangle  node {\LARGE $y$} (-1.5,6.75);
\draw [ color={rgb,255:red,1; green,0; blue,0} , fill={rgb,255:red,254; green,255; blue,255}] (-1,10.75) rectangle  node {\LARGE $w$} (-0.25,10);
\draw [short] (-1,11) .. controls (-1.25,10.5) and (-3.25,10.5) .. (-2.25,12);

\draw [ fill={rgb,255:red,0; green,0; blue,0} , line width=0.2pt ] (-0.75,11.25) circle (0cm);

\end{circuitikz}
}%
    \caption{If there is a path from $S_{uv}$ to the segment between $x$ and $u$ (not including $u$), then there is a rooted $K_4$-minor (left). If there is a path from $S_{uv}$ to $y$, then there is a rooted $K_4$-minor (right).}
    \label{fig:StoXU}
\end{figure}
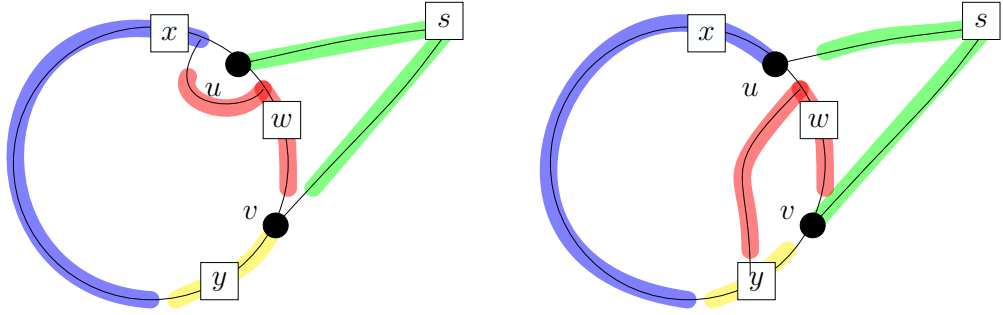

Since $w$ and $x$ are not the subdivision vertex of $u$ and $v$, $S_{uv}$ and $V(C)\setminus (\{u,v\}\cup S_{uv})$ both contain at least one vertex from $G$. Therefore, there is a path in $G\setminus \{u,v\}$ from $S_{uv}\cap V(G)$ to $(C\setminus S_{uv})\cap V(G)$, since $G$ is 3-connected. This means that one of the paths mentioned above exists, contradicting the fact that $G'$ has no $(S\cup R)$-rooted $K_4$-minor. This contradiction shows that $C$ must contain all vertices of $S$, as desired. 

Hence, one of the two arcs of $C$ between $u,v\in C$ does not contain any vertices from $S\cup R$, and we may replace this by the path via $s$ to get a cycle $C$ that additionally passes through $s$. Reiterating this argument at most $n$ times gives the cycle passing through $R \cup S$. Since each iteration takes at most $m^{1+o(1)}$ time, the lemma follows.
\end{proof}

\section{Description of algorithm}
This section is organized as follows.
In \Cref{ssec:virtual-problem}, we define an auxiliary problem called \textsc{Virtual Edge Steiner Tree}.
In \Cref{ssec:dynamic-programming} we present a dynamic programming algorithm for \textsc{Virtual Edge Steiner Tree} for the case that all roots lie on a cycle. This cycle will be provided to us by Lemma~\ref{lem:comb_lemma}.
\Cref{ssec:preprocess} discusses preprocessing steps.
Finally, in \Cref{ssec:main-algorithm}, we describe our algorithm for \textsc{Virtual Edge Steiner Tree}.

\subsection{Virtual edges} \label{ssec:virtual-problem}
We first slightly generalize the \textsc{Steiner Tree} problem to a problem that we call \textsc{Virtual Edge Steiner Tree} in order to facilitate a recursive approach. 

This problem is defined as follows. The input to \textsc{Virtual Edge Steiner Tree} is 
\begin{itemize}
    \item a graph $G=(V,E)$ with weights $w_e$ for $e \in E$,
    \item a set of terminals $T = \{ t_1, \dots, t_{k} \}\subseteq V$,
    \item a set of ``virtual edges'' $E^*$ with weights $w_{e^*} : \{u,v,\mathsf{d},\mathsf{c}\}\to \mathbb{R}_{\geq 0}$ for $e^*=\{u,v\} \in E^*$.
\end{itemize}
We will assume that $E\cap E^*=\emptyset$.
\smallskip

A \textit{solution} to the \textsc{Virtual Edge Steiner Tree} problem is a (virtual) edge set $S \subseteq E \cup E^*$ such that $T \subseteq V(S)$ and $V(S) \cap V(e^*) \ne \emptyset$ for each virtual edge $e^* \in E$. Since the edge weights are non-negative we can assume $S$ is a tree.
For a (virtual) edge set $S$, the \textit{cost} $c_{e^*}(S)$ of a virtual edge $e^*=\{u,v\}\in E^*$ with respect to $S$ is given by
\begin{itemize}
    \item $w_{e^*}(u)$ if $u\in V(S)$ and $v\not\in V(S)$, representing the cost when $u$ is included in the solution and $v$ is not,
    \item $w_{e^*}(v)$ if $v\in V(S)$ and $u\not\in V(S)$, representing the cost when $v$ is included in the solution and $u$ is not,
    \item $w_{e^*}(\mathsf{c})$ if $e\in S$,  representing the cost when $u,v$ are both included in the solution and connected ``via the virtual edge'', that is, in the part of the graph we ``forgot about'',
    \item $w_{e^*}(\mathsf{d})$ if $u,v\in V(S)$ and $e\not \in S$, representing the cost when $u,v$ are both included in the solution but are not connected ``via the virtual edge''.
\end{itemize}
The \textit{total cost} of a solution $S$ is defined as
\begin{equation}
\label{eq:cost}
c(S)=\sum_{e\in S}w_e+\sum_{e^*\in E^*}c_{e^*}(S).
\end{equation}
The purpose of virtual edges with cost functions is to enable our algorithm to handle 2-cuts $\{u,v\}$ effectively. When the algorithm identifies such a cut, it recursively solves the four subproblems for the smaller side first, encoding the costs in $w_{\{u,v\}}$.
It is crucial to ensure that the solutions on both sides of the cut agree on the inclusion of vertices $u$ and $v$ in the final solution. Additionally, if both vertices are included, the algorithm must determine which side will contain the path connecting them.

In an instance of \textsc{Virtual Edge Steiner Tree}, we will assume that there is no terminal incident with a virtual edge.
If there is a terminal $t$ and a virtual edge $tt'$, then we assign the virtual edge weight $w_{tt'}(t') = \infty$, and remove $t$ from the terminal set. 

It will be convenient for the description of the dynamic program to view a root as either a singleton set consisting of a vertex (if it is a terminal), or an edge (if it is a virtual edge). Therefore, for an instance $(G$, $w$, $T$, $E^*$, $\{w_{e^*}\}_{e^* \in E^*})$, we define the set of \emph{roots} as $R = \{\{t\}:t\in T\} \cup E^*$.
We say the instance has no rooted $K_4$-minor if the graph $G^*$ has no $R^*$-rooted $K_4$-minor, where
$G^*$ is the graph obtained from the graph $(V, E \cup E^*)$ by subdividing each edge in $E^*$ once, and $R^*$ is the union of $T$ and the subdivision vertices of $G^*$.

Note that \textsc{Virtual Edge Steiner Tree} can be reduced to \textsc{Steiner Tree} as follows:

\begin{observation} \label{obs:virtual-edge}
    There is a reduction from an instance of \textsc{Virtual Edge Steiner Tree} with $k$ terminals and $\ell$ virtual edges to $4^{\ell}$ instances of \textsc{Steiner Tree} with at most $k + 2 \ell$ terminals.
\end{observation}
\begin{proof}
    Recall that there are four cases for each virtual edge: (i) $u$ is covered but $v$ is not, (ii) $v$ is covered but $u$ is not, (iii) both $u$ and $v$ are covered and the edge $uv$ is part of the solution, and (iv) $u$ and $v$ are covered and the edge $uv$ is not part of the solution.
    For each virtual edge, we delete it from the graph, and do one of the following.
    For case (i), add $u$ as a terminal and delete $v$.
    For case (ii), add $v$ as a terminal and delete $u$.
    For case (iii), contract $u$ and $v$ into a single vertex and add it as a terminal.
    For case (iv), add both $u$ and $v$ as terminals.
    This results in a \textsc{Steiner Tree} instance with at most $k+2\ell$ terminals.
    The minimum cost of solutions among these instances plus the virtual edge weights is the minimum cost of the original instance of \textsc{Virtual Edge Steiner Tree}.
\end{proof}

\subsection{Dynamic programming when roots lie on a cycle} \label{ssec:dynamic-programming}
In this subsection, we present a polynomial-time algorithm for the \textsc{Virtual Edge Steiner Tree} problem when all roots lie on a cycle.

\begin{lemma}
An instance of \textsc{Virtual Edge Steiner Tree} on an $n$-vertex graph without rooted $K_4$-minor that has a cycle passing through all roots, can be solved in $O(n^4)$ time.
\end{lemma}
\begin{proof}
    We may also assume that there are at least five roots, since otherwise the problem can be efficiently solved via \Cref{obs:virtual-edge}.
    Let $R = \{ r_1, \dots, r_k \}$ denote the set of roots.
    We renumber the indices so that $r_1, \dots, r_k$ appear on $C$ in the order of their indices. Recall that no terminal is incident with a virtual edge, and since the virtual edges lie on $C$, each non-terminal is incident with at most two virtual edges. 
    Indices will be considered modulo $k$ (identifying $r_k$ with $r_0$).
    We use the shorthand $R[a, b]$ to denote the set of roots $\{ r_c \in R : c \in [a, b] \}$.
    In particular, $R[a, b] = \{ r_a, \dots, r_k, r_1, \dots, r_b \}$ for $a > b$.

    We say a tree $\mathcal{T}=(V(\mathcal{T}),E(\mathcal{T}))$ is a \textit{partial interval solution} if $V(\mathcal{T})\subseteq V(G)$, $E(\mathcal{T})\subseteq E\cup E^*$ and the set $R(\mathcal{T})$ of roots $r\in R$ for which $r \cap V(\mathcal{T}) \ne \emptyset$ forms an interval $R[a,b]$ for some $a,b\in [k]$.

    We determine the status of virtual edges incident with the partial interval solution in the same way as we do for a solution, and we also assign a cost in a way similar to (\ref{eq:cost}), but we now ignore virtual edges that are not incident with $S$
    \[
    c(\mathcal{T})=\sum_{e\in E(\mathcal{T})\cap E} w_e+\sum_{e^*\in E(\mathcal{T})\cap E^*} c_{e^*}(E(\mathcal{T})).
    \]
    Note that each solution $S\subseteq E\cup E^*$ gives rise to a partial interval solution $\mathcal{T}=(V(E)\cup V(E^*), E\cup E^*)$, since  $R(\mathcal{T})=[k]$ is always an interval.
    
    We say a partial interval solution $\mathcal{T}$ is a \textit{minimal solution} if $R(\mathcal{T})=R$ (so it contains all terminals and contains at least one endpoint per virtual edge) and it is minimal in the sense that there is no strict subtree which also has this property. The last assumption is needed for technical reasons since the weights could also be zero. 
    
    We compute a dynamic programming table $\mathsf{DP}$ with entries $\mathsf{DP}[a,b,v,s_a,s_b]$, where $a,b \in [k]$, $v \in V(G) \setminus T$, $s_a \in r_a \cup \{\mathsf{d},\mathsf{c}\}$  and $s_b \in r_b \cup \{\mathsf{d},\mathsf{c}\}$.
    If $r_a$ is a terminal, then the table may ignore $s_a$: we will store the same value irrespective $s_a$. 
    The same applies to~$s_b$. 
    
    The table entry $\mathsf{DP}[a, b, v,s_a,s_b]$ represents the minimum cost of a partial interval solution $\mathcal{T}$ with $R(\mathcal{T}) \supseteq R[a, b]$ and $v \in V(\mathcal{T})$, where the status of $r_a$ and $r_b$ are indicated by $s_a$ and $s_b$, respectively. In fact, the minimum will not go over all partial interval solutions, but only those that can be built in a specific manner as defined next. This means the lower bound is always true.

   We define a collection $\mathcal{B}$  of partial interval solutions via the following set of rules.  
    \begin{itemize}
        \item[\textcolor{gray}{\textbf{R1}}] Every partial interval solution $(\{v\},\emptyset)$ consisting of a single vertex  with $v\in V(G)$ is in $\mathcal{B}$. 
        \item[\textcolor{gray}{\textbf{R2}}] Every partial interval solution $\mathcal{T}$ obtained from a partial interval solution $\mathcal{T}'
    \in\mathcal{B}$ by adding a new vertex $v\in V(G)$ not incident to any roots with an edge in $E$ is in $\mathcal{B}$.
        \item[\textcolor{gray}{\textbf{R3}}] Every partial interval solution $\mathcal{T}$ obtained from two partial interval solutions $\mathcal{T}_1,\mathcal{T}_2\in \mathcal{B}$ by `gluing on $v$' is in $\mathcal{B}$ when $V(\mathcal{T}_1)\cap V(\mathcal{T}_2)=\{v\}$ and all points in the intersection of the intervals $R(\mathcal{T}_1)$ and $R(\mathcal{T}_2)$ are endpoints of both intervals.
    \end{itemize}
    Note that by the definition of partial interval solution, $R(\mathcal{T})$ is an interval $R[a,b]$ for each $\mathcal{T}\in \mathcal{B}$. 
    The following claim shows our dynamic program may restrict itself to partial interval solutions.
    \begin{claim}\label{clm:family}
    Let $\mathcal{T}$ be a minimal solution. Then $\mathcal{T}\in \mathcal{B}$.
    \end{claim}
    \begin{proof}
       We want to root our tree $\mathcal{T}$ in  a well-chosen vertex $v$ and eventually prove the claim by induction. For $w\in V(\mathcal{T})$, let  $\mathcal{T}_w$ denote the subtree of $\mathcal{T}$ containing all descendants of $w$. The choice of $v$ needs to ensure that $R(\mathcal{T}_w)\neq [k]$ for all children $w$ of $v$. By minimality, $R(\mathcal{T}_w)\neq \emptyset$ for all children $w$ of $v$.

      If there is at least one terminal $r_a$, then we can choose $v=r_a$, since $r_a$ will then be missing from $R(\mathcal{T}_w)$ for all $w\neq v$. We show in the remainder of this paragraph that such a vertex $v$ can also be chosen when no terminal exists. If there is a virtual edge incident to a single vertex $u$, we can choose $v=u$ similar to the terminal case. Moreover, if there is a vertex $u$ incident with two virtual edges, then we may choose $v=u$.
      Otherwise, all roots come from virtual edges that have both endpoints present in $\mathcal{T}$ and each vertex is incident with at most one edge. 
      We give each vertex of $G$ incident with a virtual edge weight $1$.  Since $\mathcal{T}$ is a tree, we can root it in a `balanced separator': a vertex $v$ such that the total weight in each component of $\mathcal{T}-v$ is at most half the total weight of $\mathcal{T}$. The only way $R(\mathcal{T}_w)$ can then be $[k]$ for $w$ a child of $v$, is when all roots are coming from edges incident with one vertex in $\mathcal{T}_w$ and one outside. Since there are at least five roots, we can choose roots $r_a,r_b,r_c,r_d$ consecutive on the cycle, and connect $r_a$ to $r_c$ using only internal vertices from $V(\mathcal{T}_w)$ and $r_b$ to $r_d$ using only internal vertices not in $V(\mathcal{T}_w)$, finding a rooted $K_4$-minor with  \Cref{lemma:cycle-to-k4}, a contradiction.

        In this paragraph, we show that $R(\mathcal{T}_w)$ is an interval. We are done if $|R(\mathcal{T}_w)| \ge k - 1$.
        So we assume that $|R(\mathcal{T}_w)| \le k - 2$ and $w\neq v$. Suppose towards a contradiction that $R(\mathcal{T}_w)$ is not an interval. This means there exist $a < b < c < d$ in $[k]$ such that $r_a , r_{c} \in R(\mathcal{T}_w)$ and $r_b, r_d \not \in R(\mathcal{T}_w)$ (or $b<c<d<a$, but this case is analogous).
        This means there is a path between (the roots of) $r_a$ and $r_{c}$ via vertices in $V(\mathcal{T}_w)$. Moreover, $r_b$ and $r_d$ are not in $R(\mathcal{T}_w)$, and so they are either part of the tree $\mathcal{T}'$ obtained from $\mathcal{T}$ by removing $\mathcal{T}_w$, or incident to a vertex in this tree. This provides a path between (the roots of) $r_b$ and $r_d$ using vertices not in $V(\mathcal{T}_w)\cup R(\mathcal{T}_w)$.
        Thus, by \Cref{lemma:cycle-to-k4}, we obtain an $\{ r_a, r_b, r_c, r_d \}$-rooted $K_4$-minor, a contradiction.

        In this paragraph we show that $\mathcal{S}_w=(\{w\},\emptyset)\in \mathcal{B}$ for each $w\in V(\mathcal{T})$. 
        For this,  we need to show that $R(\mathcal{T})$ forms an interval. Recall that no terminal is incident with a virtual edge, and since the virtual edges lie on $C$, each non-terminal is incident with at most two virtual edges which are then also consecutive on the cycle. This means $|\mathcal{S}_w|\leq 2$ and that $R(\mathcal{S}_w)$ is an interval. Since $\mathcal{S}_w$ consists of a single vertex, it is now added to $\mathcal{B}$ in \textcolor{gray}{\textbf{R1}}.
        
        The remainder of the proof shows that $\mathcal{T}_w\in \mathcal{B}$ for each $w\in V(\mathcal{T})$ by induction on $|V(\mathcal{T}_w)|$. We already proved (a stronger variant of) the case when the size is 1 above.     
         Suppose $\mathcal{T}_w\in \mathcal{B}$ has been shown for all $w$ with $|V(\mathcal{T}_w)|\leq \ell$ for some $\ell\geq 1$ and now assume $|V(\mathcal{T}_w)|= \ell+1$. 

        We first show a property that we need in both cases considered below. Let $w'$ be a child of $w$ and suppose that $r_x\in R(\mathcal{S}_w)\cap R(\mathcal{T}_{w'})$. Then $r_x$ must correspond to a virtual edge $ww_1$ with $w_1\in V(\mathcal{T}_{w'})$.  We show that $r_x$ is an endpoint of the interval $R(\mathcal{T}_{w'})$. Recall that $R(\mathcal{T}_{w'})\neq [k]$ by choice of $v$. If $r_x$ is not an endpoint, then there are vertices $r_a,w_1,r_x,w,r_b,r_d$ appearing in order on the cycle for some $r_a,r_b\in R(\mathcal{T}_{w'})$ and $r_d\not\in R(\mathcal{T}_{w'})$. In particular, $r_x$ can be connected to $r_d$ via $w$ using vertices outside of $V(\mathcal{T}_{w'})$ and we can connect $r_a$ and $r_b$ using internal vertices in $V(\mathcal{T}_{w'})$.         
            We then apply \Cref{lemma:cycle-to-k4} to obtain an $\{ r_{a}, r_{x}, r_{b}, r_{d} \}$-rooted $K_4$-minor, a contradiction. 

        We now turn to the proof that $\mathcal{T}_w\in \mathcal{B}$. We already proved its roots form an interval.
        Suppose that $w$ has a single child $w'$. By the induction hypothesis, $\mathcal{T}_{w'}$ is a partial interval solution. If $w$ is not incident with any roots, then $\mathcal{T}_w$ can be obtained from $\mathcal{T}_{w'}$ with \textcolor{gray}{\textbf{R2}}. 
        
        If $w$ is incident with at least one root, we show we can obtain $\mathcal{T}_w$ by gluing the partial interval solution $\mathcal{S}_w=(\{w\},\emptyset)$ with $\mathcal{T}_{w'}$. To apply \textcolor{gray}{\textbf{R3}}, we need to show all points in the intersection $R(\mathcal{S}_w)\cap R(\mathcal{T}_{w'})$ are endpoints of the corresponding intervals. 
        \begin{itemize}
            \item If $w$ is a terminal $r_a$, then  it is not incident with any virtual edges and hence there are no intersection points. 
\item Suppose that $w$ is not a terminal, but it is incident with either one or two virtual edges $r_{x_1}=ww_2$ and $r_{x_2}=ww_2$. This means $R(\mathcal{S}_w)=\{r_{x_1},r_{x_2}\}$ and so these are the only possible intersection points. We proved above that if $r_{x_j}$ exists and is in $R(\mathcal{T}_{w'})$, it must be an endpoint of that interval (for $j=1,2$). Moreover, each root in $R(\mathcal{S}_w)$ is automatically an endpoint of the interval since the size is 1 or 2. 
     \end{itemize}
        So we may assume $w$ has at least two children.
        Let $w_1,\dots,w_d$ denote the children of $w$ in $\mathcal{T}$  and for each $i\in [d]$, let $\mathcal{T}_i=\mathcal{T}_{w_i}$ denote the subtree of $\mathcal{T}$ containing all descendants of $w_i$. We already showed that for each $i\in [d]$, $R(\mathcal{T}_i)$ is an interval $R[a_i,b_i]$ for some $a_i,b_i\in [k]$. 

        The union of these intervals will be $R(\mathcal{T}_w)\setminus \{w\}$. We now describe when $R[a_i,b_i]$ and $R[a_j,b_j]$ can intersect for some $i\neq j$. Every intersection point comes from a virtual edge $uu'$ with $u\in \mathcal{T}_i$ and $u'\in \mathcal{T}_j$, and in particular each root appears in at most two of the $\mathcal{T}_i$. 
        We show that when $r_x\in R[a_i,b_i]\cap R[a_j,b_j]$, we must have $x\in \{a_i,b_i\}$. Indeed, if not, we select $r_c\not\in R[a_i,b_i]$ and now $r_c,r_{a_i},r_x,r_{b_i}$ appear consecutively on the cycle, with a path between $r_{a_i}$ and $r_{b_i}$ contained within $V(\mathcal{T}_i)$ and a path between $r_x$ and $r_{c}$ outside of $V(\mathcal{T}_i)$ (via $u'$), yielding a contradiction via an $\{r_c,r_{a_i},r_x,r_{b_i}\}$-rooted $K_4$-minor by  \Cref{lemma:cycle-to-k4}.

        Moreover, we already argued that $R(\mathcal{S}_w)\cap R[a_i,b_i]$ can only intersect in $a_i$ or $b_i$ (and in fact, at most one of those two by minimality arguments). Combined, this means that there is a way to renumber such that
        \[
        \bigcup_{i=1}^j R[a_i,b_i]\cup R(\mathcal{S}_{w}) \text{ and }\bigcup_{i=j+1}^d R[a_i,b_i]\cup R(\mathcal{S}_{w}) 
        \]
        form intervals intersecting only in their endpoints. Here $j$ is chosen so that either $w=r_{j}$ if $w$ is terminal, and $w$ is between $r_{b_j}$ and $r_{a_{j+1}}$ on the cycle if one (or both) of those is a virtual edge incident to $w$, and otherwise $R(\mathcal{S}_w)=\emptyset$ and does not affect whether the sets are intervals.
     \end{proof}

    Our dynamic program will ensure that each partial interval solution in $\mathcal{B}$ is considered by following rules \textcolor{gray}{\textbf{R1}}-\textcolor{gray}{\textbf{R3}}.
    Since these rules are also easily seen to result in a valid Steiner tree, the final output therefore will be optimal by Claim~\ref{clm:family}.

    \proofsubparagraph*{Base case.} 
    All entries are set to $\infty$ initially.
    We compute all entries $\mathsf{DP}[a,b,v,s_a,s_b]$ corresponding to a partial interval solution consisting of a single vertex via \Cref{obs:virtual-edge}. That is, for each vertex $v$, 
    \begin{itemize}
        \item if $v$ is a terminal $r_a$, we set $\mathsf{DP}[a,a,v,s_a,s_a']=0$;
        \item if $v$ has exactly one incident virtual edge $r_a$, we set $\mathsf{DP}[a,a,v,v,v]=w_{r_a}(v)$.  
        \item if $v$ has incident virtual edges $r_a$ and $r_{a+1}$, we set $\mathsf{DP}[a,a+1,v,v,v]=w_{r_a}(v)+w_{r_{a+1}}(v)$. 
    \end{itemize} 

    \proofsubparagraph*{Adding a vertex.}
    We first implement \textcolor{gray}{\textbf{R2}}: adding a vertex not incident to any roots. Suppose that $v$ is not incident with any roots.
    For any edge $uv\in E$, we add the rule
    \[
    \mathsf{DP}[a, b, v,s_a,s_b] \leftarrow \min(\mathsf{DP}[a, b, u ,s_a,s_b]+w_{uv},\mathsf{DP}[a, b, v,s_a,s_b]).
    \]
     Next, we implement \textcolor{gray}{\textbf{R3}}, `gluing on $v$', which has multiple cases because we also need to properly keep track of the costs of the virtual edges. 
\proofsubparagraph*{Gluing two solutions without intersection.}
   We first consider the case in which the intervals of the roots are disjoint, in which case the partial interval solutions that we are gluing do not have any virtual edges between them. 
    When $a_2=b_1+1$ and $a_1\not\in [a_2,b_2]$, we set
        \[
        \mathsf{DP}[a_1, b_2, v,s_{a_1},s_{b_2}]
        \leftarrow \min( \mathsf{DP}[a_1,b_2,v,s_{a_1},s_{b_2}],\mathsf{DP}[a_1, b_2, v, s_{a_1}, s_{b_2}]+\mathsf{DP}[a_2, b_2, v,s_{a_2},s_{b_2}]).
        \]
        
    \proofsubparagraph*{Gluing two solutions with one intersection.}
    Now suppose that the intervals overlap in exactly one point. 

    When $a_2=b_1$ with $r_{a_2}$ a terminal and $a_1\not\in [a_2+1,b_2]$, we set 
\[
\mathsf{DP}[a_1, b_2, v,s_{a_1},s_{b_2}] \leftarrow \min(\mathsf{DP}[a_1, b_2, v,s_{a_1},s_{b_2}], \mathsf{DP}[a_1, b_1, v,s_{a_1},s_{b_1}]+\mathsf{DP}[a_2, b_2, v,s_{a_2},s_{b_2}]).
\]
We may forget about $s_{a_2},s_{b_1}$ since the `status' of the terminal is irrelevant to us.

When $r_{a_2}=uu'$ is a virtual edge,  we need to consider various options, depending on the status this edge used to be in for both solutions, and what final state we want it to be, and whether $a_1=b_1$ and/or $a_2=b_2$. 

We start with when we want to make the status $\mathsf{d}$. We need to ensure $u$, $u'$ are part of the solution, and the solutions must overlap elsewhere in some vertex $v$. So when
$a_2=b_1$, 
$a_1\not\in [a_2+1,b_2]$, we set
    \begin{align*}
        \mathsf{DP}[a_1,b_2, v,x_1,x_2] \leftarrow \min(&\mathsf{DP}[a_1, b_2, v,x_1,x_2], \\
        &\mathsf{DP}[a_1,b_1, v,s_{a_1},u]+\mathsf{DP}[a_2, b_2, v,u',s_{b_2}]-w_{uu'}(u)-w_{uu'}(u')+w_{uu'}(\mathsf{d})), \\
        &\mathsf{DP}[a_1, b_1, v,s_{a_1},u']+\mathsf{DP}[a_2, b_2, v,u,s_{b_2}]-w_{uu'}(u)-w_{uu'}(u')+w_{uu'}(\mathsf{d})),\\
        \intertext{where we require $x_1=s_{a_1}$ if $a_1\neq b_1$ and $x=\mathsf{d}$  when $a_1=b_1$; and  we require $x_2=s_{b_2}$ if $a_2\neq b_2$ and $x_2=\mathsf{d}$ when $a_2=b_2$.   
        In the last line, for example, we use that the first partial interval solution has paid the cost $w_{uu'}(u')$ (as recorded by the status) and the second $w_{uu'}(u)$. We allow here to replace the cost by $w_{uu'}(\mathsf{d})$, since the solutions can be merged via the vertex $v$. If $a_1=b_1$ or $a_2=b_2$, we will keep the status $\mathsf{d}$ recorded at the relevant endpoint.}
\intertext{We also give the option to `connect' via the virtual edge $r_{a_2}$. When $r_{a_2}=uu'$ is a virtual edge,  $a_2=b_1$ 
$a_1\not\in [a_2+1,b_2]$, we set}
 \mathsf{DP}[a_1,b_2, u,x_1,x_2] \leftarrow \min(&\mathsf{DP}[a_1, b_2, u,x_1,x_2], \\
        &\mathsf{DP}[a_1, b_1, u,s_{a_1}, u]+\mathsf{DP}[a_2, b_2, u', u',s_{b_2}] - w_{uu'}(u)-w_{uu'}(u')+w_{uu'}(\mathsf{c})),\\
 \mathsf{DP}[a_1,b_2, u',x_1,x_2] \leftarrow \min(&\mathsf{DP}[a_1, b_2, u',x_1,x_2], \\
        &\mathsf{DP}[a_1, b_1, u,s_{a_1}, u]+\mathsf{DP}[a_2, b_2, u', u',s_{b_2}] - w_{uu'}(u)-w_{uu'}(u')+w_{uu'}(\mathsf{c})),
        \intertext{where again, we require $x_i=s_{a_i}$ if $a_i\neq b_i$ and $x_i=\mathsf{c}$ otherwise.}
        \intertext{To allow the edge to be in status $u$ (or $u'$, analogously), we also add}
         \mathsf{DP}[a_1,b_2, v,x_1,x_2] \leftarrow \min(&\mathsf{DP}[a_1, b_2, v,x_1,x_2], \\
        &\mathsf{DP}[a_1, b_1, v,s_{a_1}, u']+\mathsf{DP}[a_2, b_2, v, u,s_{b_2}] - w_{uu'}(u'),\\
        &\mathsf{DP}[a_1, b_1, v,s_{a_1}, u]+\mathsf{DP}[a_2, b_2, v, u,s_{b_2}] - w_{uu'}(u),\\
         &\mathsf{DP}[a_1, b_1, v,s_{a_1}, u]+\mathsf{DP}[a_2, b_2, v, u',s_{b_2}] - w_{uu'}(u')),
        \intertext{where again, we require $x_i=s_{a_i}$ if $a_i\neq b_i$ and $x_i=\mathsf{c}$ otherwise.}
    \end{align*}
    \proofsubparagraph{Gluing two solutions with two intersections}
    It remains now to glue intervals with two intersection points to a final solution. The rules for this are analogous to the previous case, but now the compatibility is checked at both endpoints, $a_1=b_2$ and $a_2=b_1$.
 
\medskip

    We apply the rules in a bottom-up fashion. The final output is the minimum over all entries $\mathsf{DP}[a,b,v,s_a,s_b]$ with $R[a,b] = R$. 

    \proofsubparagraph*{Running time.}
    There are $O(n^3)$ entries, each of which takes $O(n)$ time to compute.
    Thus, the total running time is $O(n^4)$.
\end{proof}

\subsection{Preprocessing} \label{ssec:preprocess}
Throughout, the algorithm works on a simple graph $G^*$ (that is, at most one edge per pair of vertices). We can always obtain this via the following \emph{edge pruning} steps.
\begin{enumerate}
    \item If there are two edges $e$ and $e'$ containing the same two endpoints with $w_e \le w_{e'}$, then delete $e'$.
    \item If there are two virtual edges $e_1^*$ and $e_2^*$ over the same endpoints $u$ and $v$, then delete $e_1^*$ and $e_2^*$ add a new virtual edge $e^*$ between $u$ and $v$ with the virtual weight defined as follows:
    \begin{itemize}
    \item $w_{e^*}(u) = w_{e_1^*}(u) + w_{e_2^*}(u)$. 
    \item $w_{e^*}(v) = w_{e_1^*}(v) + w_{e_2^*}(v)$. 
    \item $w_{e^*}(\mathsf{c}) = \min(w_{e_1^*}(\mathsf{c}) + w_{e_2^*}(\mathsf{d}), w_{e_1^*}(\mathsf{d}) + w_{e_2^*}(\mathsf{c}))$. 
    \item $w_{e^*}(\mathsf{d}) = w_{e_1^*}(\mathsf{d}) + w_{e_2^*}(\mathsf{d})$.
    \end{itemize}
    \item If there is an edge $e$ and virtual edge $e^*$ over the same endpoints say $u$ and $v$, then delete~$e$, and update the virtual weight of $e^*$ by $w_{e^*}(\mathsf{c}) = \min(w_{e^*}(\mathsf{c}), w_{e^*}(\mathsf{d}) + w_e)$.
\end{enumerate}

We briefly discuss the correctness of the preprocessing steps above, though all of them follow directly from the definition of \textsc{Virtual Edge Steiner Tree}.
Firstly, note that we exactly remove all multiple edges (and preserve which vertices have at least one edge between them). For (1), an edge with a larger weight is never included in an optimal solution so can be removed. For (2), at most one of the two virtual edges will be used to `connect' the vertices, but we still need to pay the cost for both even if they are not `included' in the solution.
Hence the new edge has the weight defined as the sum of the weight of two virtual edges for the cases $u$, $v$, and $\mathsf{d}$.
For the case $\mathsf{c}$, we may assume that exactly one of $e_1^*$ and $e_2^*$ is included into the solution. 
Lastly, for (3), we update $w_{e^*}$ to take into account that it may be cheaper to connect via the edge $e$ in the scenario that $u$ and $v$ needs to be connected.
That is, if there is an edge $e$ and virtual edge $e^*$ such that $w_{e^*}(\mathsf{c}) \geq  w_{e^*}(\mathsf{d}) + w_e$, then we update $w_{e^*}(\mathsf{c})$ to $w_{e^*}(\mathsf{d}) + w_e$. In the other scenarios an optimal solution would never add the edge $e$.

\smallskip

We define a preprocessing procedure as follows to ensure that every biconnected component has at least one root and that every triconnected component has at least two roots.

We say a vertex set $A\subseteq V$ is \textit{incident with a root} if $A\cap T\neq \emptyset$ or there is a virtual edge incident with a vertex of $A$.
\begin{enumerate}
    \item Perform edge pruning.
    \item If there is a cut vertex $v$ such that $G - v$ has a connected component $A$ that is not incident with any roots, then delete $A$ from the graph. Return to 1.
    \item If there is a 2-cut $\{ u, v \}$ such that $G[V\setminus \{u,v\}]$ has a connected component $A$ that is not incident with a root, then delete $A$ from the graph and add an edge $uv$, where the weight $w_{uv}$ equals $\dist_{G[A \cup \{ u, v \}]}(u, v)$. Return to 1.
    % \carla{I added step 1 and the return to this after each step; I merged the last two steps into the step below and took care of the boundary case in which the cut has a virtual edge, when we need to be careful to remove it. I also ensured the write-up now specified how the virtual edge relates to $A$ (in particular when it is not in $A$ but incident with it).}
    \item If there is a 2-cut $\{ u, v \}$ such that $G[V\setminus \{u, v\}]$ has a connected component $A$ that is incident with exactly one root, we compute the weights for a new virtual edge $\{u,v\}$ by solving four \textsc{Virtual Edge Steiner Tree} instances using Observation \ref{obs:virtual-edge}.
    \begin{itemize}
        \item $w_{uv}(u)$ is the cost of the instance induced by $A \cup \{ u \}$ with $u$ as an additional terminal.
        \item $w_{uv}(v)$ is the cost of the instance induced by $A \cup \{v \}$ with $v$ as an additional terminal.
        \item $w_{uv}(\mathsf{c})$ is the cost of the instance induced by $A \cup \{ u, v \}$ with $\{ u, v \}$ as terminals but where we do not include a virtual edge between $u$ and $v$ if it had one.
        \item $w_{uv}(\mathsf{d})=\min(w_{uv}(u),w_{uv}(v))$.
    \end{itemize}
     We add a new virtual edge $\{ u, v \}$ with the costs as defined above and remove all vertices from $A$.
(If there was already a virtual edge between $u$ and $v$, we keep it and it will be dealt with during the cleaning.) Return to 1.
\end{enumerate}

For (2), note that since $v$ is a cut vertex, there is always an optimal solution that has empty intersection with  $A$.

For (3), if $u$ and $v$ are connected in an optimal solution via $A$, then this will be done via a shortest path between $u$ and $v$.

For (4), there are four cases to consider. In the first three cases, we enforce $u$, $v$ or both $u$ and $v$ are contained in the optimal Steiner tree by making them terminals. In the last, we note that the solution on $A$ is allowed to go either via $u$ or via $v$. Note that the instances that we solve to define the costs all have at most 3 roots (virtual edges or terminals) and so can be solved in polynomial time in terms of $|A|$ using Observation \ref{obs:virtual-edge}. 

\subsection{Our algorithm} \label{ssec:main-algorithm}

\begin{algorithm}[t]
	\caption{A polynomial-time algorithm for \textsc{Virtual Edge Steiner Tree}. We assume that the input graph is connected and that there is no rooted $K_4$-minor in the instance.}
	\label{main-alorithm}  \label{algo:main}
	\begin{algorithmic}[1]	

	\Procedure{Algo}{$(G,T,E^*,w_e)$}
    \State Apply preprocessing from \Cref{ssec:preprocess}
	\State \algorithmicif\ $|T| + |E^*| = O(1)$ \algorithmicthen\ \Return Solve by \Cref{obs:virtual-edge} and Dreyfus-Wagner.
     \If {$G^*$ is not 2-connected}
      \State Let $v$ be a cut vertex that separates $G$ into $A$ and $B$.
      \State \Return  \textsc{Algo}($G[A \cup \{ v \}]$, $(T \cap A) \cup \{ v \}$, $E^* \cap \binom{A}{2}$, $w$) \\
      \hspace{8em} + \textsc{Algo}($G[B \cup \{ v \}]$, $(T \cap B) \cup \{ v \}$, $E^* \cap \binom{B}{2}$, $w$).
    \ElsIf {$G^*$ is not 3-connected}
        \State Let $\{ u, v \}$ be a cut that separates $G$ into $A$ and $B$ with $|A| \le |B|$.
            \State Let $e$ be a virtual edge connecting $u$ and $v$.
            \State $w_{e}(u) \leftarrow$ \textsc{Algo}($G[A \cup \{ u \}], (T \cap A) \cup \{ u \}, E^* \cap \binom{A}{2}, w$)
            \State $w_{e}(v) \leftarrow$ \textsc{Algo}($G[A \cup \{ v \}], (T \cap A) \cup \{ v \}, E^* \cap \binom{A}{2}, w$)
            \State $w_{e}(\mathsf{c}) \leftarrow$ \textsc{Algo}($G[A \cup \{ u, v \}], (T \cap A) \cup \{ u, v \}, E^* \cap \binom{A}{2}, w$)
            \State $w_{e}(\mathsf{d}) \leftarrow$ \textsc{Algo}($G[A \cup \{ u, v \}] + uv, (T \cap A) \cup \{ u, v \}, E^* \cap \binom{A}{2}, w$)]
        \State \Return \textsc{Algo}($G[B \cup \{ u, v \}], T, E^* \cup \{ e \}, w$)
    \Else
        \State \Return the result for 3-connected case.
    \EndIf
	\EndProcedure
	\end{algorithmic}
\end{algorithm}

See \Cref{algo:main} for the outline of the algorithm.
Our algorithm maintains the invariant that there is no rooted $K_4$-minor over the terminals $T$ and virtual edges $E^*$.
We first apply the preprocessing steps in \Cref{ssec:preprocess}.
If this results in a graph with $O(1)$ roots, we solve the problem by reducing to \textsc{Steiner Tree} with $O(1)$ terminals, which can then be solved using the Dreyfus-Wagner algorithm~\cite{DBLP:journals/networks/DreyfusW71} in $O(1)$ time.

\subparagraph*{$G^*$ is not 2-connected.}

Let $v$ be a cut vertex separating $G$ into $A$ and $B$.
(Here, $A$ and $B$ each may contain multiple connected components of $G - v$.)
We recursively solve two instances induced on vertex sets $A \cup \{ v \}$ and $B \cup \{ v \}$ where $v$ is an additional terminal. 
To see why this recursion is correct, note that $A$ and $B$ each contains at least one terminal or virtual edge by the preprocessing steps.
So, any Steiner tree must cover $v$ as well, and thus
it is obtained by ``gluing'' the two solutions on the vertex $v$.
The total cost equals the sum of the two recursive instance costs.

We claim that no rooted $K_4$-minor is introduced in the recursive calls.
As $v$ is a cut vertex, there is a path from $v$ to each vertex in $B$. Moreover, by the preprocessing, $B$ contains a root (there is either a virtual edge in $B\cup \{v\}$ or a terminal in $B$). This means that any rooted $K_4$-minor in the instance on $A\cup \{v\}$ with $v$ as terminal also gives rise to a rooted $K_4$-minor in $G$. 

\subparagraph*{$G^*$ is not 3-connected.} 

We first find a 2-cut $\{ u, v \}$ that separates the graph into $A$ and~$B$ with $|A|\leq |B|$.
Our algorithm considers four cases based on whether $u$ and $v$ are covered by the solution.
    \begin{itemize}
        \item We solve the subproblem over $G[A \cup \{ u, v \}]$ with $u$ as an additional terminal, to account for the case in which $u$ is covered but $v$ is not. 
        \item We solve the subproblem over $G[A \cup \{ u, v \}]$ with $v$ as an additional terminal, to account for the case in which $v$ is covered but $u$ is not. 
        \item We solve the subproblem over $G[A \cup \{ u, v \}]$ with $u$ and $v$ as additional terminals. This accounts for the case that both $u$ and $v$ are covered, and they are connected through $A$.
        \item We solve the subproblem over $G[A \cup \{ u, v \}]$ with $u$ as terminal and the edge $\{u,v\}$ added (in $E$) with cost $w(\{u,v\})=0$. This accounts for the case in which both $u$ and $v$ are covered, and they are not connected through $A$. 
    \end{itemize}
    Next, we solve a single instance on vertex set $B \cup \{u,v\}$ where $\{u,v\}$ is added as virtual edge (in $E^*$).
    The cost of the virtual edge $\{ u, v \}$ is determined by the previous calculations 

    We verify that no rooted $K_4$-minor is introduced in the recursive calls.
    The first two recursive calls (corresponding to the cases $u$ and $v$) are analogous to the scenario where $G^*$ has a cut vertex.
    For the third and fourth calls (cases $\mathsf{c}$ and $\mathsf{d}$), assume for contradiction that there exists a rooted $K_4$-minor over $T \cup \{u, v\}$ in $G[A \cup \{u, v\}] + uv$.
    Our preprocessing steps ensure two roots $r, r' \in T \cup E^*$. Since there is no cut vertex, there are two disjoint paths connecting $\{r, r'\}$ to $\{u, v\}$.
    Using these paths, we can construct a rooted $K_4$-minor in $G$, which is a contradiction. For the final recursive call on $B \cup \{u, v\}$, the argument for the case where $G^*$ has a cut vertex again shows the non-existence of a rooted $K_4$-minor.

    \subparagraph*{$G^*$ is 3-connected.}
    If $G^*$ is 3-connected, then we can find a polynomial time a cycle going through all roots by \Cref{lem:comb_lemma} presented in Section~\ref{sec:comb_lemma},
    and \textsc{Virtual Edge Steiner Tree} can be solved as described in \Cref{ssec:dynamic-programming}

\paragraph*{Runtime analysis.}

When we solve an instance with $G^*$ being $3$-connected, we already showed the running time is at most $c'n^4$ for some constant $c'>0$, assuming $n \geq n_0$ for some constant $n_0$. 
After making $c'$ larger if needed, we can also assume all additional steps in the algorithm such as finding a 2-cut and preprocessing take at most $c'n^3$ on inputs of $n \geq n_0$ vertices.

We let $T(n)$ denote the maximal running time on an input graph on $n$ vertices and show that $T(n)\leq cn^4$ for some constant $c>0$ by induction on $n$. We will choose $c$ such that $c\geq c'$ and $T(n)\leq cn^4$ when $n\leq n_0$.

When $G^*$ is not 2-connected, we find a cut vertex $v$ splitting the graph into $A$ and $B$ and recursively solve two instances on $A\cup \{v\}$ and $B\cup \{v\}$ and obtain the cost from there. So with $|A|=i$, we find the running time is at most $T(i+1)+T(n-i)+c'n^3$.

When $G^*$ is 2-connected, but not 3-connected, we find a 2-cut $\{u,v\}$ that separates $G^*$ into $A$ and $B$ with $|A|\leq |B|$, and therefore $|A|\leq (n-2)/2$. With $i=|A|$, we recursively solve four instances of size at most $i$ and one instance of size $|B|\leq n-i$, resulting in a running time of at most
\[
\max_{1\leq i\leq n/2} T(n - i) + 4T(i + 2)+c'n^3.
\]
Applying the inductive hypothesis, we find $T(n-i)\leq c(n-i)^4$ and $T(i+2)\leq c(i+2)^4$ for all $1\leq i\leq n/2$. 
The function $f(i) = c (n - i)^{4} + 4c(i + 2)^4 +c'n^{3}$ is convex within the domain $1 \le i \le n/2$ because the second derivative
$f''(i)$ (with respect to $i$)
is positive. Therefore, the maximum value is attained either at $i = 1$ or $i = n/2$. Evaluating these points, 
\begin{align*}
    f(1) &= c(n-1)^4 + 4c(3^4)+ c'n^3 \\
    % = cn^4-4cn^3+ 6cn^2- 4cn+ c+4c(3^4)+c'n^3, so need 4c>c'
    f(n/2) &= c(n/2)^4+4c(n/2+2)^4 + c'n^3,
\end{align*}
we see both are less than $cn^4$ when $n$ is sufficiently large since $c\geq  c'$.
Hence, we conclude that $T(n) = O(n^4)$.

\bibliographystyle{abbrv}
\bibliography{refs}

\end{document}